\documentclass[11pt,a4paper,english]{amsart}
\usepackage[utf8]{inputenc}
\usepackage[english]{babel} 

\usepackage{amsmath} 
\usepackage{amsthm} 
\usepackage{graphicx} 
\usepackage{xcolor} 
\usepackage{amsfonts}
\usepackage[biblabel]{cite}
\usepackage{bm} 
\usepackage[hidelinks]{hyperref} 
\usepackage{amssymb} 
\usepackage{tikz-cd} 
\usepackage{amscd}
\usepackage{etoolbox}
\patchcmd{\thmhead}{(#3)}{#3}{}{}
\usepackage{enumitem}
\usepackage{breqn} 
\usepackage{faktor} 
\usepackage{xfrac} 

\usepackage{algorithm}
\usepackage{algpseudocode}
\algnewcommand\algorithmicinput{\textbf{Input:}}
\algnewcommand\Input{\item[\algorithmicinput]}
\algnewcommand\algorithmicoutput{\textbf{Output:}}
\algnewcommand\Output{\item[\algorithmicoutput]}

\DeclareMathOperator{\supp}{supp} 
\DeclareMathOperator{\ev}{ev} 

\DeclareMathOperator{\PRS}{PRS}
\DeclareMathOperator{\PRM}{PRM}
\DeclareMathOperator{\RM}{RM}

\DeclareMathOperator{\wt}{wt}

\DeclareMathOperator{\RS}{RS}

\newcommand{\F}{{\mathbb{F}}}
\newcommand{\fq}{\mathbb{F}_q}

\newcommand{\PM}{{\mathbb{P}^{m}}}

\newcommand{\A}{{\mathbb{A}}}

\newcommand{\DRM}{D^{\mathbb{A}}_d}
\newcommand{\DRMdd}{D^{\mathbb{A}}_{d-1}}
\newcommand{\DPRM}{D^{\mathbb{P}}_d}
\newcommand{\DPRMq}{D^{\mathbb{P}}_{d-(q-1)}}
\newcommand{\fgood}{{f_{good}}}
\newcommand{\fbad}{{f_{bad}}}
\newcommand{\fgoodo}{{f_{good}^0}}
\newcommand{\fbado}{{f_{bad}^0}}

\newcommand{\cbad}{{c_{bad}}}
\newcommand{\cgood}{{c_{good}}}
\newcommand{\evp}[1]{\ev^{\mathbb{P}}_{{#1}}}
\newcommand{\eva}[1]{\ev^{\mathbb{A}}_{{#1}}}
\newcommand{\etadm}{{\eta_d(m)}}
\newcommand{\comp}{{\theta^{\A}(m)}}

\usepackage{mathtools} 
\DeclarePairedDelimiter\abs{\lvert}{\rvert}%
\DeclarePairedDelimiter\norm{\lVert}{\rVert}%

\makeatletter
\let\oldabs\abs
\def\abs{\@ifstar{\oldabs}{\oldabs*}}
\let\oldnorm\norm
\def\norm{\@ifstar{\oldnorm}{\oldnorm*}}
\makeatother

\newtheorem{thm}{Theorem}[section]
\newtheorem{prop}[thm]{Proposition}
\newtheorem{cor}[thm]{Corollary}
\newtheorem{lem}[thm]{Lemma}
\theoremstyle{definition}
\newtheorem{defn}[thm]{Definition} 

\newtheorem{rem}[thm]{Remark} 
 
\newtheorem{ex}[thm]{Example}

\title{Recursive decoding of projective Reed-Muller codes}
\author{Rodrigo San-José}
\curraddr{
Department of Mathematics\\ Virginia Tech\\ Blacksburg, VA USA. \textit{Previous address:} IMUVA-Mathematics Research Institute, Universidad de Valladolid, 47011 Valladolid (Spain).
}
\email{rsanjose@vt.edu}

\thanks{This work was supported in part by the following grants: Grant PID2022-138906NB-C21 funded by MICIU/AEI/10.13039/501100011033 and by ERDF/EU, and FPU20/01311 funded by the Spanish Ministry of Universities.}



\begin{document}
\maketitle

\begin{abstract}
We give a recursive decoding algorithm for projective Reed-Muller codes making use of a decoder for affine Reed-Muller codes. We determine the number of errors that can be corrected in this way, which is the current highest for decoders of projective Reed-Muller codes. We determine the degrees for which we can decode up to the error correction capability of these codes, and we compute the order of complexity of the algorithm, given by that of the chosen decoder for affine Reed-Muller codes. We also show how the decoder behaves when increasing the number of variables and the size of the field. 
\end{abstract}

\section{Introduction}
Projective Reed-Muller (PRM) codes are a family of evaluation codes introduced in \cite{lachaudPRM}. They are obtained by evaluating multivariate homogeneous polynomials in the projective space, and their basic parameters were completely determined in \cite{sorensen,ghorpadeMinimumweightPRM}. Their affine counterpart, affine Reed-Muller (RM) codes, have been widely studied. We know the basic parameters of RM codes \cite{kasamiRM,delsarteRM}, but also many additional properties \cite{pellikaanGHWRM,capacityRM,steaneRM,polarRM,RMtheoryandalgorithms,bookRM,kasamiWeightStructureRM}. There has been some recent work in the direction of closing the gap in knowledge between RM and PRM codes \cite{beelenGHWPRM,sanjoseRecursivePRM,sanjoseHullsPRM,sanjoseHullvariationPRM,ghorpadeMinimumweightPRM,sanjoseSSCPRM,sanjoseSSCPRS,kaplanHullsPRM,songHullParamPRM}. However, one of the main aspects required for a family of codes to be useful in practice is to have an efficient decoding algorithm. RM codes admit several efficient decoding algorithms, most of them derived from the fact that they can be seen as algebraic geometry codes \cite{sakata1,sakata2,sakataplusvoting,duursmaMajority,fengraoMajority}. For PRM codes, the only specific decoder was given in \cite{decodingRMP}. This decoder is efficient in terms of complexity, but does not decode up to the error correction capability of PRM codes. 

In \cite{sanjoseRecursivePRM}, a recursive construction for PRM codes is given. This recursive construction is similar to a $(u,u+v)$ construction, which is known to have an efficient decoding algorithm \cite{decodingMPC,hernandoDecodingMPC2}. Moreover, affine RM codes also admit a recursive construction which has been used to obtain decoding algorithms \cite{abbeRecursiveDecodingRM,dumerRecursiveDecodingRM}. In this paper, we propose an algorithm based on the recursive construction from \cite{sanjoseRecursivePRM} to decode PRM codes. 

In Section \ref{s:preliminaries} we provide the necessary background. In this section, we also introduce a particular ordering for the points of the projective space, which is a key element for the recursive structure of PRM codes. We define an integer $\eta_d(m)$, which will determine the error correction capability of our algorithm, and we study some of its properties. Since in general $\eta_d(m)$ is only a lower bound for the minimum distance of the corresponding PRM code, the decoder may not decode up to the error correction capability of the code for all possible degrees, and we determine the degrees for which it does. In Section \ref{s:decodingalg}, we study how to decode PRM codes recursively using decoders for RM codes. First, in Subsection \ref{ss:prs} we show how this idea gives a decoding algorithm for projective Reed-Solomon (PRS) which decodes up to their error correction capability, using a decoder for affine Reed-Solomon (RS) codes. Then, in Subsection \ref{ss:prm2} we explain how the algorithm works for the case of PRM codes over the projective plane. This already features most of the aspects of the general recursive algorithm, but it is easier to grasp and serves as an intermediate step between PRS codes and the general case of PRM codes. We also provide examples to help understand some of the underlying problems we have to solve to successfully decode PRM codes. In Subsection \ref{ss:PRMgeneral}, we give the recursive decoding algorithm, Algorithm \ref{alg:1}, and we prove that it can decode any number of errors lower than $\eta_d(m)/2$. By the definition of $\eta_d(m)$, this means that we can always decode at least the same number of errors as the algorithm presented in \cite{decodingRMP}, but in general, we will be able to correct more. In fact, we study for which degrees Algorithm \ref{alg:1} can decode up to the error correction capability of PRM codes. This is done in Section \ref{s:analysis}, where we also consider an improved algorithm, Algorithm \ref{alg:2}, which can additionally decode some particular error patterns even if the weight of the error exceeds the aforementioned error correction capability of Algorithm \ref{alg:1}. In Section \ref{s:complexity} we compute the order of complexity of Algorithms \ref{alg:1} and \ref{alg:2}, which is the same as the complexity order of the chosen decoder for RM codes. Finally, in Section \ref{s:conclusion} we mention some future avenues of research.

\section{Preliminaries}\label{s:preliminaries}
Let $\fq$ be the finite field of $q$ and let $m$ be a positive integer. We consider the projective space $\mathbb{P}^m$ over $\fq$, and we denote by $p_j$ the number of points in $\mathbb{P}^j$, i.e., $p_j=\abs{\mathbb{P}^j}=\frac{q^{j+1}-1}{q-1}$. Throughout this work, we will fix the standard representatives for $\mathbb{P}^m$, which are the representatives with the leftmost nonzero coordinate equal to 1. The set $P^m$ whose elements are these representatives can be regarded as a subset of the affine space:
$$
P^{m}:=\left(\{1\}\times \F_{q}^m\right) \cup \left(\{0\}\times \{1\}\times \F_{q}^{m-1}\right)\cup\cdots\cup \{(0,\dots,0,1)\}\subset\mathbb{A}^{m+1} .
$$
Note that we can construct $P^m$ recursively:
\begin{equation}\label{descomposicion}
P^m=\left(\{1\}\times \F_{q}^m\right) \cup \left( \{0\}\times P^{m-1}\right),
\end{equation}
where we use the convention $P^0:=\{1\}$. In what follows, we will assume a specific ordering of the points in $P^m$. First, let $\xi\in \fq$ be a primitive element. We order the elements of $\fq$ in terms of the powers of $\xi$, that is
$$
\fq=\{\xi^0,\xi^1,\xi^2,\dots,\xi^{q-2},0\}.
$$
By using Equation \ref{descomposicion}, we see that this also gives an ordering of the points of $P^1$ as follows:
$$
P^1=\{ (1,\xi^0),(1,\xi^1),(1,\xi^2),\dots,(1,\xi^{q-2}), (1,0), (0,1)\}.
$$
Now we notice that 
\begin{equation}\label{descomposicion2}
\F_{q}^m=P^{m-1}\cup \xi\cdot P^{m-1}\cup\cdots\cup \xi^{q-2}\cdot P^{m-1}\cup \{(0,\dots,0)\}.
\end{equation}
Indeed, given a point $Q$ in $\F_{q}^m\setminus \{(0,\dots,0\}$, its leftmost nonzero coordinate is equal to $\xi^r$ for some $0\leq r\leq q-2$, which implies that $Q\in \xi^r \cdot P^{m-1}$. In particular, we have
$$
\fq^2=P^1\cup \xi\cdot P^1 \cup \xi^2\cdot P^1\cdots \cup \xi^{q-2}\cdot P^1\cup \{(0,0)\}.
$$
Thus, the ordering given to $\fq$ determines a particular ordering of the points of $P^1$, which also determines an ordering of the points of $\fq^2$. Recursively, we fix the ordering of the elements of $P^j$ and $\fq^j$ from the ordering we have given for $\fq$, for any $j$, and we will always assume that we are using these orderings in what follows. This is particularly important for practical implementations of the algorithms presented in this work, since they are recursive and they assume the ordering of the points is compatible with that recursion. 

\begin{ex}\label{ex:P2}
Let $q=4$ and $\F_4=\{1,a,a+1,0\}$, with $a^2=a+1$. This gives the ordering of the points (choosing $\xi=a$ as primitive element):
$$
\begin{aligned}
P^2=\{(1, 1, 1),
 (1, 1, a),
 (1, 1, a + 1),
 (1, 1, 0),
 (1, 0, 1),
 (1, a, a),
 (1, a, a + 1),
 (1, a, 1),
 (1, a, 0), \\
 (1, 0, a),
 (1, a + 1, a + 1),
 (1, a + 1, 1),
 (1, a + 1, a),
 (1, a + 1, 0),
 (1, 0, a + 1),
 (1, 0, 0), \\
 (0, 1, 1),
 (0, 1, a),
 (0, 1, a + 1),
 (0, 1, 0),
 (0, 0, 1)\}.
\end{aligned}
$$
\end{ex}

Fix $m>0$. For $j\leq m$, we define the evaluation map
$$
\evp{j}:\fq[x_{m-j},\dots,x_m] \rightarrow \fq^{p_j},\:\: f\mapsto \left(f(Q_1),\dots,f(Q_n)\right)_{Q_i \in P^j}.
$$
Let $d$ be a positive integer, and let $\fq[x_0,\dots,x_m]_d$ be the set of homogeneous polynomials of degree $d$. Then $\PRM_d(q,j):=\evp{j}(\fq[x_{m-j},\dots,x_m]_d)$ is the PRM code of degree $d$ over $P^j$. We may use the notation $\PRM_d(j)$ if there is no confusion about the field. We have introduced the definition of $\evp{j}$ and $\PRM_d(j)$ because it will be useful for the decoding algorithm, but we will mainly focus on $\PRM_d(m)$ in what follows. For $m=1$, we obtain PRS codes (sometimes called doubly extended Reed-Solomon codes), which are MDS codes with parameters $[q+1,d+1,q-d+1]$. We use the notation $\PRS_d=\PRM_d(1)$. The following result from \cite{sorensen} gives the basic parameters of these codes (also see \cite{ghorpadeMinimumweightPRM,sorensenGapMinD} for the minimum distance).

\begin{thm}\label{T:paramPRM}
The code $\PRM_d(m)$, for $1\leq d\leq m(q-1)$, is an $[n,k]$-code with 
$$
\begin{aligned}
&n=\frac{q^{m+1}-1}{q-1},\\
&k=\sum_{t\equiv d\bmod q-1,0<t\leq d}\left( \sum_{j=0}^{m+1}(-1)^j\binom{m+1}{j}\binom{t-jq+m}{t-jq}   \right).\\
\end{aligned}
$$
For the minimum distance, we have
$$
\wt(\PRM_d(m))=(q-\mu )q^{m-\nu-1}, \text{ where } \;
d-1=\nu(q-1)+\mu, \;0\leq \mu <q-1.
$$
\end{thm}

Let $d>0$ and let $I(\mathbb{P}^m)$ be the vanishing ideal of $\mathbb{P}^m$, i.e., the ideal generated by the homogeneous polynomials that vanish at all the points of $\PM$. From \cite{mercier} we have that
$$
I(\mathbb{P}^m)=\langle \{x_i^qx_j-x_ix_j^q,\; 0\leq i <j\leq m \} \rangle.
$$
From the definition of $\evp{m}$ and $\PRM_d(m)$, it is clear that $\fq[x_0,\dots,x_m]_d/I(\mathbb{P}^m)_d\cong \PRM_d(m)$. Therefore, we can regard the vectors of $\PRM_d(m)$ as classes of polynomials in $\fq[x_0,\dots,x_m]_d/I(\mathbb{P}^m)_d$. By fixing a representative for each class, we obtain a basis for $\fq[x_0,\dots,x_m]_d/I(\mathbb{P}^m)_d$, which also gives a basis for $\PRM_d(m)$ via $\evp{m}$. Let 
$$
M_d(j):=\{x^\alpha=x_{m-j}^{\alpha_{m-j}}\cdots x_m^{\alpha_m}:\abs{\alpha}=d,\alpha_{m-j}>0,0\leq \alpha_{m-i}\leq q-1, i<j\leq m\},
$$
for $j=0,1,\dots,m$, and $M_d:=\bigcup_{j=0}^m M_d(j)$. Then we have that the classes of the monomials of $M_d$ form a basis for $\fq[x_0,\dots,x_m]_d/I(\mathbb{P}^m)_d$ (see \cite{projectivefootprint}).

\begin{rem}\label{r:homogeneizar}
Let $d\geq 1$ and let $x^\alpha\in M_{d'}(j)$, for some $j\leq m$ and $d'<d$ with $d'\equiv d\bmod q-1$. Then $\alpha_{m-j}>0$ and
$$
x^{\alpha'}:=x_{m-j}^{\alpha_{m-j}+\lambda (q-1)}x_{m-j+1}^{\alpha_{m-j+1}}\cdots x_m^{\alpha_m} \in M_d(j),
$$
where $\lambda=(d-d')/(q-1)$. Note that we have $\evp{m}(x^\alpha)=\evp{m}(x^{\alpha'})$. This also proves that $\PRM_{d'}(m)\subset \PRM_d(m)$. 
\end{rem}

Let $j\leq m$. We will also use affine RM codes extensively, which we denote by $\RM_d(q,j)$ (or $\RM_d(j)$ if there is no confusion about the field). We consider the evaluation map
$$
\eva{j}:\fq[x_{m-j+1},\dots,x_m]\rightarrow \fq^{q^j},\:\: f\mapsto \left(f(Q_1),\dots,f(Q_{q^j})\right)_{Q_i \in \fq^j}.
$$
If we denote by $\fq[x_{m-j+1},\dots,x_m]_{\leq d}$ the polynomials of degree less than or equal to $d$, we have $\RM_d(j):=\eva{j}(\fq[x_{m-j+1},\dots,x_m]_{\leq d})$. As with PRM codes, we will mainly consider the code $\RM_d(m)$. For $m=1$ we obtain Reed-Solomon codes of length $q$, for which we use the notation $\RS_d=\RM_d(1)$. From \cite{delsarteRM,kasamiRM} we have the following result about the basic parameters of RM codes.

\begin{thm}\label{T:paramRM}
The code $\RM_d(m)$, for $0\leq d\leq m(q-1)$, is an $[n,k]$-code with 
$$
\begin{aligned}
&n=q^m,\\
&k=\sum_{t=0}^d \sum_{j=0}^m (-1)^j\binom{m}{j}\binom{t-jq+m-1}{t-jq}.\\
\end{aligned}
$$
For the minimum distance, we have
$$
\wt(\RM_d(m))=(q-\mu )q^{m-\nu-1}, \text{ where } \;
d=\nu(q-1)+\mu, \;0\leq \mu <q-1.
$$
\end{thm}

\begin{rem}\label{r:dminRMYPRM}
As a consequence of Theorems \ref{T:paramPRM} and \ref{T:paramRM}, we have
$$
\wt(\PRM_d(m))=\wt(\RM_{d-1}(m)). 
$$
\end{rem}

Similarly to PRM codes, we have the isomorphism $\fq[x_1,\dots,x_m]_{\leq d}/I(\A^m)\cong \RM_d(m)$, where
$$
I(\A^m)=\langle x_i^q-x_i,\; 1\leq i \leq m \rangle.
$$
For $1\leq j\leq m$, let
$$
A(j):=\{x^\alpha \in \fq[x_{m-j+1},\dots,x_m] : 0\leq \alpha_i\leq q-1,\; 1\leq i \leq m\}.
$$
Then every class in $\fq[x_1,\dots,x_m]/I(\A^m)$ has one representative (and only one) in $A(m)$. Thus, whenever we consider a polynomial $f\in \fq[x_1,\dots,x_m]$, we denote by $f \bmod I(\A^m)$ its representative in $\fq[x_1,\dots,x_m]/I(\A^m)$ whose monomials are contained in $A(m)$. We say that $f$ is reduced modulo $I(\A^m)$ if it is expressed in terms of the monomials in $A(m)$. We can now give the following result, which is related to Remark \ref{r:homogeneizar} and will be used later. We use the following notation: given a monomial $x^\alpha=x_0^{\alpha_0}\cdots x_m^{\alpha_m}\in \fq[x_0,\dots,x_m]$, we denote by $x^\alpha(1,x_1,\dots,x_m)$ the monomial $x_1^{\alpha_1}\cdots x_m^{\alpha_m}$. 

\begin{lem}\label{l:monomiosenA}
Let $d\geq q$ and let $x^\alpha\in M_d(a),x^\beta \in M_d(b)$ be distinct monomials with $a\geq b$ and assume we have $x^\alpha(1,x_1,\dots,x_m)\equiv x^\beta(1,x_1,\dots,x_m)\bmod I(\A^m)$. Then $x_0\mid x^\alpha$ and $x_0\nmid x^\beta$, and there exists a monomial $x^{\beta'}\in M_{d-(q-1)}$ such that $\evp{m}(x^\beta)=\evp{m}(x^{\beta'})$.
\end{lem}
\begin{proof}
If $x_0\nmid x^\alpha$ and $x_0\nmid x^\beta$, we have that $x^\alpha,x^\beta\in \fq[x_1,\dots,x_m]_d$ and $\eva{m}(x^\alpha)=\eva{m}(x^\beta)$. We also have $\evp{m-1}(x^\alpha)=\evp{m-1}(x^\beta)$ (we can see $P^{m-1}\subset \fq^m$ for this case). Thus, $\evp{m}(x^\alpha)=\evp{m}(x^\beta)$, a contradiction since they are distinct elements of $M_d$. 

If $x_0\mid x^\alpha$ and $x_0\mid x^\beta$, then $a=b=m$, and $x^\alpha(1,x_1,\dots,x_m)\equiv x^\beta(1,x_1,\dots,x_m)\bmod I(\A^m)$ implies that all the powers of the variables of $x^\alpha$ and $x^\beta$, different from $x_0$, are the same. Since they are of the same degree, this means $x^\alpha=x^\beta$, a contradiction. 

If $x_0\mid x^\alpha$ and $x_0\nmid x^\beta$ (recall $a\geq b$), we have that $x^\beta(1,x_1,\dots,x_m)=x^\beta$ is of degree $d$, and $x^\alpha(1,x_1,\dots,x_m)$ is of degree $<d$. Thus, in order to have 
$$
x^\alpha \equiv x^\beta(1,x_1,\dots,x_m) \bmod I(\A^m),
$$
the reduction of $x^\beta$ modulo $I(\A^m)$ has to be different from $x^\beta$, which implies that the degree of at least one variable in $x^\beta$ is higher than $q-1$. Then
$$
x^{\beta'}=x_{m-b}^{\beta_{m-b}-(q-1)}\cdots x_m^{\beta_m}
$$
satisfies the properties of the statement (see Remark \ref{r:homogeneizar}). 
\end{proof}

We can obtain polynomials in $\fq[x_{m-j},\dots,x_m]_{d}$ using polynomials in $\fq[x_{m-j+1},\dots,x_m]_{\leq d}$ via homogenization. We define
$$
h_d^j:\fq[x_{m-j+1},\dots,x_m]_{\leq d} \rightarrow \fq[x_{m-j},\dots,x_m]_{d}, \; f\mapsto x_{m-j}^{d}f\left(\frac{x_{m-j+1}}{x_{m-j}},\dots,\frac{x_{m}}{x_{m-j}}\right).
$$
Note that $h_d^j(A(j))\subset M_d(j)$. Now we introduce additional notation related to homogenizations, which will be needed for Algorithm \ref{alg:1}. 

\begin{defn}\label{def:badmonomials}
Let $1\leq d \leq m(q-1)$. Let $x^\alpha\in \fq[x_1,\dots,x_m]$ be a monomial. We say that $x^\alpha$ is \textit{bad} if $0<\abs{\alpha}<d$ and $\abs{\alpha} \equiv d\bmod q-1$. We say that $x^\alpha$ is \textit{good} if it is not bad. Given a polynomial $f \in \fq[x_1,\dots,x_m]_{\leq d}$, we denote by $\fbad$ the sum of the terms of $f$ corresponding to bad monomials. Let $\fgood=f-\fbad$. We denote by $(\fgood)_d$ and $(\fgood)_{\leq d-1}$ the sum of the terms of degree $d$ and degree less than or equal to $d-1$ of $\fgood$, respectively.
\end{defn}

If $d\leq q-1$ and $f \in \fq[x_1,\dots,x_m]_{\leq d}$, it is clear that $\fbad=0$. For $d\geq q$, the notion of bad monomials is related to having several possible homogenizations. This will be illustrated in Examples \ref{ex:differenthom} and \ref{ex:differenthom2}.

From \cite{sanjoseRecursivePRM} we have a recursive construction for PRM codes, which also involves RM codes. Note that this construction relies on the ordering of the elements of $P^m$ and $\fq^m$ that we have chosen. 

\begin{thm}\label{T:recursive}
Let $1\leq d\leq m(q-1)$ and let $\xi$ be a primitive element in $\F_{q}$. We have the following recursive construction:
{\small
$$
\PRM_d(m)=\{(u+v_{\xi,d},v)\mid  u\in \RM_{d-1}(m), v\in \PRM_d(m-1)\},
$$}
where $v_{\xi,d}:=v\times \xi^d v\times \cdots\times \xi^{(q-2)d}v\times \{0\}  =(v,\xi^d v,\xi^{2d}v,\dots,\xi^{(q-2)d}v,0)$.
\end{thm}

As a particular case of the implications of the recursive construction for the generalized Hamming weights of PRM codes given in \cite[Thm. 7]{sanjoseRecursivePRM}, this construction gives a bound $\eta_d(m)$ for the minimum distance of PRM codes, which will determine the error correction capability of the recursive decoding algorithm that we will develop in what follows. Let $d-1=\nu (q-1)+\mu$, with $0\leq \mu <q-1$. We denote
\begin{equation}  
\begin{aligned}\label{eq:eta0}
\etadm:= \sum_{i=0}^{m-\nu-1} \wt(\RM_d(m-i))+1
\end{aligned}
\end{equation}

\begin{lem}\label{l:eta0}
Let $d-1=\nu (q-1)+\mu$, with $1\leq d\leq m(q-1)$ and $0\leq \mu <q-1$. Then
$$
\etadm=(q-\mu)q^{m-\nu -1}-\mu \frac{q^{m-\nu-1}-1}{q-1}=\wt(\PRM_d(m))-\mu \frac{q^{m-\nu-1}-1}{q-1}.
$$
\end{lem}
Therefore,
$$
\wt(\PRM_d(m))\geq \etadm,
$$
and we have
$$
\wt(\PRM_d(m))= \etadm
$$
if and only if $\mu=0$ or $\nu=m-1$. 
\begin{proof}
If $\mu <q-2$, then $d=\nu(q-1)+(\mu+1)$, where $\mu+1<q-1$. Therefore, by Theorem \ref{T:paramRM} we have
$$
\begin{aligned}
\etadm&=\sum_{i=0}^{m-\nu-1}(q-\mu-1)q^{m-\nu-i-1}+1\\
&=(q-\mu)q^{m-\nu -1}+(q-\mu-1)\sum_{i=1}^{m-\nu-1}q^{m-\nu-i-1}-q^{m-\nu-1}+1\\
&=(q-\mu)q^{m-\nu -1}+(q-\mu-1)\frac{q^{m-\nu-1}-1}{q-1}-q^{m-\nu-1}+1\\
&=(q-\mu)q^{m-\nu -1}-\mu \frac{q^{m-\nu-1}-1}{q-1}=\wt(\PRM_d(m))-\mu \frac{q^{m-\nu-1}-1}{q-1},
\end{aligned}
$$
where the last equality follows from Theorem \ref{T:paramPRM}. If $\mu=q-2$, then $d=(\nu+1)(q-1)$. To use Theorem \ref{T:paramRM}, we consider $\nu'=\nu+1$ and $\mu'=0$. However, it is straightforward to check that the formula for the minimum distance of $\RM_d(m)$ gives the same value for $\nu'$ and $\mu'$ as it does for $\nu''=\nu$ and $\mu''=q-1$, i.e., both give the value $q^{m-\nu -1}$. Thus, the previous computation still holds in this case.
\end{proof}

\begin{cor}\label{c:etaq-1}
Let $q\leq d \leq m(q-1)$. We have that
$$
\etadm=\eta_{d-(q-1)}(m-1).
$$
\end{cor}
\begin{proof}
Let $d-1=\nu(q-1)+\mu$, with $0\leq \mu \leq q-1$. Note that $d-(q-1)-1=(\nu-1)(q-1)+\mu$, and, by Lemma \ref{l:eta0}, we have 
$$
\begin{aligned}
\eta_{d-(q-1)}(m-1)
=\wt(\PRM_{d-(q-1)}(m-1))-\mu \frac{q^{m-\nu-1}-1}{q-1}. 
\end{aligned}
$$
If we take into account Theorem \ref{T:paramPRM}, we see that $\wt(\PRM_{d-(q-1)}(m-1))=\wt(\PRM_d(m))$, and therefore $\etadm=\eta_{d-(q-1)}(m-1)$. 
    
\end{proof}

As we mentioned above, the lower bound $\etadm$ for $\wt(\PRM_d(m))$ is motivated by the recursive construction from \cite{sanjoseRecursivePRM}. Indeed, from \cite[Thm. 7]{sanjoseRecursivePRM} we have
$$
\wt(\PRM_d(m))\geq \min \{ \wt(\RM_{d-1}(m)),q \cdot \wt (\PRM_{d}(m-1)),\wt(\RM_d)+\wt(\PRM_d(m-1))\}.
$$
As a consequence of Theorems \ref{T:paramPRM} and \ref{T:paramRM}, and Remark \ref{r:dminRMYPRM}, for $d\leq (m-1)(q-1)$ we have
$$
\wt(\RM_{d-1}(m))=\wt(\PRM_d(m))=q\cdot \wt(\PRM_d(m-1).
$$
For $d> (m-1)(q-1)$, we have $\wt(\RM_{d-1}(m))\leq q\cdot \wt(\PRM_d(m-1))=q$. Thus, for $d\leq m (q-1)$ we obtain
$$
\wt(\PRM_d(m))\geq \min \{ \wt(\RM_{d-1}(m)),\wt(\RM_d)+\wt(\PRM_d(m-1))\}.
$$
As in the proof of Lemma \ref{l:eta0}, one can prove that, if $d\leq m(q-1)$, then 
$$
\wt(\RM_{d-1}(m))\geq\wt(\RM_d)+\wt(\PRM_d(m-1)),
$$
and we obtain
$$
\wt(\PRM_d(m))\geq \wt(\RM_d)+\wt(\PRM_d(m-1)).
$$
By iterating this, one gets the bound given in Lemma \ref{l:eta0}. Also, this bound can be obtained by realizing that we are evaluating a homogeneous polynomial of degree $d$ in $P^m$, and, when restricted to $\{0\}^{\ell}\times \{1\}\times \mathbb{A}^{m-\ell}\subset P^m$, it has the same evaluation as a polynomial of degree at most $d$ in $\mathbb{A}^{m-\ell}$. In this direction, the following remark will be used several times in what follows.

\begin{rem}\label{r:u+vRMd}
In Theorem \ref{T:recursive}, we always have that $u+v_{\xi,d}\in \RM_d(m)$. Indeed, if $\evp{m}(f)=(u+v_{\xi,d},v)$ for some $f\in \fq[x_0,\dots,x_m]_d$, then 
$$
\eva{m}(f(1,x_1,\dots,x_m))=\evp{m}(f(x_0,x_1,\dots,x_m))\hspace{-0.1cm}\mid_{\{1\}\times \fq^m}=(u+v_{\xi,d}),
$$
where, for $w\in \fq^{q^m}$, we denote $w\hspace{-0.1cm}\mid_{\{1\}\times \fq^m}:=(w_Q)_{Q\in \{1\}\times \fq^m}$.
\end{rem}

\section{Recursive decoding algorithm}\label{s:decodingalg}
Fix $m>0$. In this section, we show a decoding algorithm, Algorithm \ref{alg:1}, to decode $\PRM_d(m)$ if the number of errors is lower than $t_0:=\eta_d(m)/2$. We denote the corresponding decoder $\DPRM(m)$. This algorithm relies on a choice of a decoding algorithm $\DRM(j)$ for RM codes, for each $j\leq m$. To the knowledge of the author, the most efficient decoding algorithm is obtained by considering the Berlekamp-Massey-Sakata algorithm \cite{sakata1,sakata2} together with the majority voting algorithm \cite{duursmaMajority, fengraoMajority} to decode RM codes up to their error correction capability \cite{sakataplusvoting}, or even more depending on the type of errors \cite{brasCorrectionCapabilityBMS}. A simpler approach could be to just consider majority voting decoding as in \cite{hoholdtAGCodesHandbook,munueraIntroAGCodes,fengraoMajority}, which also decodes up to the error correction capability of RM codes, but with higher complexity. With respect to PRM codes, the only specific decoder known is given in \cite{decodingRMP}, in which the authors consider a decomposition of the projective space into affine spaces, while our approach relies on the recursive construction from \cite{sanjoseRecursivePRM}.

It is important to note that, given $c\in \RM_d(j)$ and $e\in \fq^{q^j}$ with $\wt(e)<\wt(\RM_d(j))/2$, we assume that $\DRM(j)(c+e)=(c,f)\in \fq^{q^j}\times \fq[x_{m-j+1},\dots,x_m]_{\leq d}$ with $c=\eva{j}(f)$, i.e., we assume that $\DRM(j)$ returns both the codeword and the polynomial in the last $j$ variables whose evaluation gives the codeword $c$. Most decoding algorithms for RM codes obtain both $c$ and $f$ simultaneously, but $f$ can also be obtained from $c$ by solving a linear system of equations (expressing $f$ in terms of the chosen basis for $\RM_d(j)$). Note that it is better to avoid this since the complexity of solving this linear system of equations is $O(k^2n)$ for a code of dimension $k$ and length $n$. This may be comparable to the complexity of some decoding algorithms for high $k$, e.g., majority decoding has complexity $O(n^3)$ \cite{hoholdtAGCodesHandbook}.  We will also assume that $\DRM(j)$ decodes up to the error correction capability of $\RM_d(j)$, which is the case for all the decoders we have mentioned for RM codes. 

As a convention, if at some point we use a decoder whose error correction capability is $0$, we assume it returns the received vector if it belongs to the corresponding code, and an error otherwise. With respect to the terminology, by the error correction capability of an $[n,k,\delta]$ code $C$ we mean $\left\lfloor \frac{\delta-1}{2}\right\rfloor$, since it can correct any number of errors lower than or equal to $\left\lfloor \frac{\delta-1}{2}\right\rfloor$, or, equivalently, any number of errors lower than $\delta/2$. For completeness, we start by briefly covering the case of PRS codes, then we show how our decoder works for $\PRM_d(2)$ to illustrate the main ideas of Algorithm \ref{alg:1}, and then we prove that the algorithm works in the general case of $\PRM_d(m)$.  

\subsection{Projective Reed-Solomon codes}\label{ss:prs}
We show now how to decode PRS codes up to their error correction capability using our recursive method (for a direct approach, see \cite{arneDecodingPRS}). Let $m=1$ and
$$
t:=\frac{\wt(\PRS_d)}{2}=\frac{q-d+1-1}{2}  = \frac{q-d}{2}=\frac{\wt(\RS_{d-1})}{2}=\frac{\eta_d(1)}{2}.
$$
Assume we want to send $c=(c_1,c_2)\in \PRS_d$, where $c_1\in \fq^{q}$, $c_2\in \fq$, and the receiver gets $r=(r_1,r_2)=(c_1+e_1,c_2+e_2)$, where $e=(e_1,e_2)$ has $\wt(e)=\wt(e_1)+\wt(e_2)< t$. First we consider $(c'_1,f'_1)=\DRM(1)(r_1)=\DRM(1)(c_1+e_1)$ (recall that, by Theorem \ref{T:recursive}, $c_1\in \RS_d$), where $c'_1\in \fq^q$ and $f'_1\in \fq[x_1]$. We define $c'_2$ as the coefficient of $x_1^d$ of the polynomial $f_1'$. If $\wt(r-c')< t$, where $c'=(c'_1,c'_2)$, then $c=c'$ and we return $(c',f_1')$. On the other hand, if $\wt(r-c')\geq  t$, then this implies that $\wt(e_1)\geq t'$, where 
$$
t'=\frac{q-d-1}{2}=t-\frac{1}{2}=\frac{\wt(\RS_d)}{2},
$$
since otherwise the previous procedure would have given $c=c'$. Thus, we have
$$
t'+\wt(e_2)\leq \wt(e_1)+\wt(e_2)<t \iff \wt(e_2)<1/2. 
$$
Therefore, $\wt(e_2)=0$ and $c_2=r_2$. Let $f\in \fq[x_0,x_1]_d$ such that $\evp{1}(f)=c$. From this, we obtain that the coefficient of $x_1^d$ in $f$ is $r_2$. Moreover, $g=(f-r_2x_1^d)(1,x_1) \in \fq[x_1]_{\leq d-1}$. Thus, 
$$
u_1=\eva{1}(g)=c_1-(r_2)_{\xi,d} \in \RS_{d-1}.
$$
We can consider $(u_1,g)=\DRMdd(1)(r_1-(r_2)_{\xi,d})=\DRMdd(1)(c_1-(r_2)_{\xi,d}+e_1)$, since $\wt(e_1)<t$, which means that it is within the error correction capability of $\RS_{d-1}$. We return $(c,f)=((u_1+(r_2)_{\xi,d},r_2), h_d^0(g)+r_2x_1^d)$ (it is easy to check that the evaluation over $P^1$ of the polynomial we are returning is equal to the codeword we return). 

In the worst case scenario, we will have to use both $\DRM(1)$ and $\DRMdd(1)$, and the order of complexity is the same as the one of $\DRM(1)$. Also, this worst case scenario happens if and only if $\wt(e)=t$ and $\wt(e_2)= 0$ (that is, there is no error in the coordinate associated to $(0,1)$). If $\wt(e_2)\neq 0$, we only need to use $\DRM(1)$. 

\subsection{Projective Reed-Muller codes over the projective plane}\label{ss:prm2}
In this subsection, we generalize the previous decoding algorithm to $\PRM_d(2)$. Let $m=2$, $\nu(q-1) < d\leq (\nu+1)(q-1)$, $0\leq \nu\leq 1$, and consider
$$
t:= \frac{\wt(\PRM_d(2))}{2}, \;\eta_d(2)=\sum_{i=0}^{1-\nu} \wt(\RM_d(m-i))+1, \text{ and } t_0:=  \frac{\eta_d(2)}{2} .
$$
We show now how to decode $\PRM_d(2)$, as long as $\wt(e)<t_0\leq t$ (the last inequality follows from Lemma \ref{l:eta0}). Let $c\in \PRM_d(2)$ and $f\in \fq[x_0,x_1,x_2]_d$ be such that $\evp{m}(f)=c$. By Theorem \ref{T:recursive}, we have
$$
c=(u+v_{\xi,d},v),\; u\in \RM_{d-1}(2), v\in \PRS_d. 
$$
Now consider $e=(e_1,e_2)$, where $e_1\in \fq^{q^2}$, $e_2\in \fq^{q+1}$, with $\wt(e)=\wt(e_1)+\wt(e_2)<t_0$, and assume we receive the vector $(r_1,r_2)=(u+v_{\xi,d}+e_1,v+e_2)$. We start with the case $\nu=0$. In that case, notice that we have 
$$
\wt(e)=\wt(e_1)+\wt(e_2)< t_0= \frac{\eta_d(2)}{2} = \frac{\wt(\RM_d(2))}{2}+\frac{\wt(\RS_d)+1}{2}.
$$
Therefore, either 
$$
\wt(e_1)< \frac{\wt(\RM_d(2))}{2}:=t_1,
$$
or 
$$
\wt(e_2)< \frac{\wt(\RS_d)+1}{2}:=t_2.
$$
If $\wt(e_1)< t_1$, then we can consider $(u+v_{\xi,d},f^0)=\DRM(2)(u+v_{\xi,d}+e_1)$ since $u+v_{\xi,d}\in \RM_d(2)$ (see Remark \ref{r:u+vRMd}) and $\wt(e_1)$ is lower than or equal to the error correction capability of $\RM_d(2)$. From $u+v_{\xi,d}$ we can obtain both $u$ and $v_{\xi,d}$, and, as a consequence, $v$. Indeed, $f^0\in \fq[x_1,x_2]$ can be written as 
$$
f^0=f^0_{\leq d-1}+f^0_d,
$$
where $f^0_{\leq d-1}\in \fq[x_1,x_2]_{\leq d-1}$ and $f^0_d \in \fq[x_1,x_2]_d$. Then we have
$$
u=\eva{2}(f^0_{\leq d-1}), \; v=\evp{1}(f^0_d), \; c=\evp{2}(h^0_d(f^0)).
$$
Note that $f(1,x_1,x_2)=f^0$ (we do not need to take into account the equations of $I(\A^2)$ since $\nu=0$ implies $d\leq q-1$), and the monomials of $f$ and $f^0$ are in bijection via $h_d^0$. 

When decoding, we do not know if $\wt(e_1)< t_1$, but we can nonetheless proceed as above and obtain a vector $c'$. If the corresponding error vector $(r-c')$ has weight less than $t$, then we conclude $c=c'$ and we finish the algorithm. If not, then we know $\wt(e_1)\geq t_1$. In that case, we have $\wt(e_2)< t_2$ and $\wt(e_2)$ is within the error correction capability of $\PRS_d$, which means that we can consider $(v,g)=\DPRM(1)(v+e_2)$ (using the algorithm from Subsection \ref{ss:prs}). If $d=q-1$, $t_2=1$, $\wt(e_2)=0$, we assume that $\DPRM(1)$ just returns the received vector and its corresponding polynomial. From $v$, we obtain $v_{\xi,d}$, and we can consider 
$$
(u,f^0_{\leq d-1})=\DRMdd(2)(r_1-v_{\xi,d})=\DRMdd(2)(u+e_1),
$$
since $\wt(e_1)\leq \wt(e_1)+\wt(e_2)=\wt(e)< t_0\leq t$, and both $\PRM_d(2)$ and $\RM_{d-1}(2)$ can correct any number of errors lower than $t$. Thus, we have already obtained $c=(u+v_{\xi,d},v)$. With respect to $f$, we just need to consider $f=h_d^0(f^0_{\leq d-1})+g$.

If $\nu=1$, then $\eta_d(2)=\wt(\RM_d(2))+1$. We can proceed as above, but now we consider $t_2:=1/2$. The only difference is that, in the case $\wt(e_1)<t_1$, when we obtain $(u+v_{\xi,d},f^0)=\DRM(2)(u+v_{\xi,d}+e_1)$, we cannot immediately deduce $u$ and $v_{\xi,d}$. We know that 
$$
f^0\equiv f(1,x_1,x_2) \bmod I(\A^2),
$$
but there might be several possible homogenizations of $f^0\in \fq[x_1,x_2]_{\leq d}$ to a polynomial in $\fq[x_0,x_1,x_2]_d$, as we see in the next example.

\begin{ex}\label{ex:differenthom}
Let $d=q$ and let $e=(e_1,e_2)\in \fq^{p_2}$ be such that $\wt(e_1)<t_1$. Assume that $f=x_0^{q-1}x_1-x_1^q$. Since $f(1,x_1,x_2)\equiv 0 \bmod I(\A^2)$, when following the procedure above we will get 
$$
\DRM(2)(u+v_{\xi,d}+e_1)=((0,\dots,0),0).
$$
Thus, we have $f^0=0$, but there are several ways to homogenize this polynomial. For example, the polynomials $x_0^{q-1}x_2-x_2^q$ and $x_0^{q-1}x_1-x_1^q+ x_0^{q-1}x_2-x_2^q$ are possible homogenizations of $f^0$ to degree $d$ which are different from $f$, and they are equivalent to $0\bmod I(\A^2)$ when setting $x_0=1$.  
\end{ex}

To obtain $f$, we need to use the information from the last $q+1$ coordinates. First we decompose $f^0$ as follows (recall Definition \ref{def:badmonomials}):
$$
f^0=\fbado +\fgoodo=\fbado+(\fgoodo)_d +(\fgoodo)_{\leq d-1}.
$$

We can now motivate Definition \ref{def:badmonomials}. Let $x^\alpha\in \supp(f^0)$. If $\abs{\alpha}=d$, then $x^\alpha\in \supp(f^0)$ implies $x^\alpha\in \supp(f)$. Moreover, if $\abs{\alpha}<d$ and $\abs{\alpha}\not\equiv d\bmod q-1$, there is only one monomial $x^\beta\in M_d$ such that $x^\beta(1,x_1,x_2)\equiv x^\alpha \bmod I(\A^2)$, which is precisely $h_d^0(x^\alpha)$ (a similar thing happens for $\abs{\alpha}=0$, even if $d\equiv 0 \bmod q-1$). Therefore, $h_d^0(x^\alpha)\in \supp(f)$. The rest of the monomials admit several different homogenizations, and that is why we call them bad monomials.

\begin{ex}\label{ex:differenthom2}
Following the setting from Example \ref{ex:differenthom}, we consider instead $f=x_0^{q-1}x_1-x_1^q+x_1^{q-1}x_2+x_0^{q-2}x_2^2+x_2^q$. In this case, 
$\DRM(2)$ will return $f^0= x_1^{q-1}x_2+x_2^2+x_2$. We have
$$
\fbado =x_2,\; \fgoodo=x_1^{q-1}x_2+x_2^2, \; (\fgoodo)_d=x_1^{q-1}x_2,\; (\fgoodo)_{\leq d-1}=x_2^2.
$$
Note that $x_2$ can be homogenized to two different monomials of degree $q$: $x_0^{q-1}x_2$ and $x_2^q$.
\end{ex}

We give the next result in general since we will also use it in the general case.

\begin{lem}\label{l:fdqminus1}
We have $\evp{m-1}((f-h_d^0(\fgoodo))(0,x_1,\dots,x_m))\in \PRM_{d-(q-1)}(m-1)$ and $\evp{m-1}((f-h_d^0(\fgoodo))(0,x_1,\dots,x_m))=\evp{m-1}((f-(\fgoodo)_d)(0,x_1,\dots,x_m))$. 
\end{lem}
\begin{proof}
Let $x^\beta \in \supp(f)$. If $x^\beta(1,x_1,\dots,x_m)\not \in \supp(f^0)$, this means that this monomial gets canceled with another monomial when reducing $\bmod \; I(\A^m)$. This also implies that $x^\beta \not \in \supp(h_d^0(\fgoodo))$ and $x^\beta \in \supp(f-h_d^0(\fgoodo))$. We claim that 
$$
\evp{m-1}((x^\beta)(0,x_1,\dots,x_m)) \in \PRM_{d-(q-1)}(m).
$$

If $x_0\nmid x^\beta$, by Lemma \ref{l:monomiosenA} there is $x^{\beta'}\in M_{d-(q-1)}$ with $\evp{m}(x^\beta)=\evp{m}(x^{\beta'})$. If $x_0\mid x^\beta$, then $\evp{m-1}((x^\beta)(0,x_1,\dots,x_m))=0$. Thus, in both cases, the claimed statement is true.

Now assume $x^\beta(1,x_1,\dots,x_m) \in \supp(f^0)$. If  $x^\beta(1,x_1,\dots,x_m) \in \supp(\fgoodo)$, since this monomial only has one possible homogenization because it is a good monomial, we have $x^\beta \not \in f-h_d^0(\fgoodo)$. If $x^\beta(1,x_1,\dots,x_m) \in \supp(\fbado)$, this monomial may have $x^\beta=h_d^0(x^\beta(1,x_1,\dots,x_m))$ or $x^\beta$ as in Remark \ref{r:homogeneizar} (we have two possible homogenizations), depending on whether $x_0\mid x^\beta$ or not. If $x_0\mid x^\beta$, we have $\evp{m-1}((x^\beta)(0,x_1,\dots,x_m))=0$. If $x_0\nmid x^\beta$, then Remark \ref{r:homogeneizar} shows that there exists $x^{\beta'}\in \fq[x_1,\dots,x_m]_{d-(q-1)}$ such that $\evp{m}(x^\beta)=\evp{m}(x^{\beta'})$ (in particular, their image by $\evp{m-1}$ is also the same since $x_0\nmid x^\beta$). 

Now let $x^\beta \not\in \supp(f)$. Then $x^\beta \in \supp(f-h_d^0(\fgoodo))$ if and only if $x^\beta \in \supp h_d^0(\fgoodo)$. If $x_0\nmid x^\beta$, then $x^\beta \in \supp(\fgoodo)_d\subset \supp(f)$, a contradiction. If $x_0\mid x^\beta$,  we have $\evp{m-1}((x^\beta)(0,x_1,\dots,x_m))=0$.

Thus, we have proved that $\evp{m-1}((f-h_d^0(\fgoodo))(0,x_1,\dots,x_m))\in \PRM_{d-(q-1)}(m-1)$. The last assertion follows from the fact that 
$$
h_d^0(\fgoodo)(0,x_1,\dots,x_m)=(\fgoodo)_d(0,x_1,\dots,x_m).
$$
\end{proof}

Let $\cgood:= \evp{1}((\fgoodo)_d)$. As a consequence of Lemma \ref{l:fdqminus1} we have
$$
r_2-\cgood=\evp{m-1}(f-(\fgoodo)_d)+e_2=\cbad +e_2,
$$
where $\cbad:=\evp{m-1}(f-(\fgoodo)_d) \in \PRM_{d-(q-1)}(1)$. Now we consider $(\cbad,f'_{bad})=\DPRMq(1)(r_2-\cgood)$. By Corollary \ref{c:etaq-1}, the error $e_2$ falls within the error correction capability of this decoder (which is the one of Subsection \ref{ss:prs}), since $\wt(e_2)<t_0$. Recovering $c$ from this is straightforward since $c=(c_1,\cgood+\cbad)=(u+v_{\xi,d},\cgood+\cbad)$. To recover $f$, first notice that $f'_{bad}$ can be homogenized to degree $d$ as in Remark \ref{r:homogeneizar}, obtaining $(f'_{bad})_d$. Let $g^0\in \fq[x_0,x_1,x_2]$ such that 
$$
g^0\equiv f^0-f'_{bad}-(\fgoodo)_d \equiv (\fgoodo)_{\leq d-1}+\fbad -f'_{bad} \bmod I(\A^2).
$$
By the reasoning above, we can assume $g^0\in \fq[x_0,x_1,x_2]_{\leq d-1}$ (when we write $\bmod I(\A^m)$, we assume we are considering reduced polynomials). Now we can recover $f$ as
$$
f=h_d^0(g^0)+(f'_{bad})_d+(\fgoodo)_d.
$$
Indeed, we have
$$
f(1,x_1,x_2)=g^0+(f'_{bad})_d+(\fgoodo)_d\equiv f^0\bmod I(\A^2),
$$
which implies $\eva{m}(f(1,x_1,x_2))=c_1$, and
$$
f(0,x_1,x_2)=(f'_{bad})_d+(\fgoodo)_d,
$$
which also implies $\evp{m-1}(f(0,x_1,x_2))=\cgood+\cbad$, that is, $\evp{m}(f)=c$. The rest of the arguments work as in the case $\nu=0$. 

\begin{ex}\label{ex:differenthom3}
Following Example \ref{ex:differenthom2}, we consider 
$$
\cgood =\evp{1}(x_1^{q-1}x_2),
$$
and $\DPRMq(1)(r_2-\cgood)$ returns $f'_{bad}=-x_1+x_2$. Then we have $(f'_{bad})_q=-x_1^q+x_2^q$ and
$$
g^0=(x_1^{q-1}x_2+x_2^2+x_2)-(-x_1+x_2)-x_1^{q-1}x_2=x_2^2+x_1.
$$
Thus, we recover
$$
f=h_q^0(x_2^2+x_1)+(-x_1^q+x_2^q)+x_1^{q-1}x_2=x_0^{q-1}x_1-x_1^q+x_1^{q-1}x_2+x_0^{q-2}x_2^2+x_2^q.
$$
\end{ex}

For $\nu=0$, in the worst case scenario we use $\DRM(2)$, $\DPRM(1)$ (which means using $\DRMdd(1)$ and $\DRM(1)$ in the worst case scenario) and $\DRMdd(2)$. On the other hand, for $\nu=1$, since $\DPRM(1)$ is trivial in that case, at most we need to use $\DRM(2)$ and $\DRMdd(2)$.

\subsection{Projective Reed-Muller codes, general case}\label{ss:PRMgeneral}

We give now the main result of this work, based on Algorithm \ref{alg:1}. For the reader interested in a more detailed description of this algorithm, it may be convenient to read Subsections \ref{ss:prs} and \ref{ss:prm2} before, since in the proof of Theorem \ref{T:alg} we reduce the number of details to avoid repetition. We denote 
$$
t:= \frac{\wt(\PRM_d(m))}{2}, \;\etadm:= \sum_{i=0}^{m-\nu-1} \wt(\RM_d(m-i))+1, \text{ and } t_0:=  \frac{\etadm}{2} .
$$

\begin{thm}\label{T:alg}
Let $d-1=\nu (q-1)+\mu$, $d\leq m(q-1)$, $0\leq \mu  <q-1$. Let $r=c+e$ be a received codeword, with $c\in \PRM_d(m)$ and $\wt(e)<t_0$. Then Algorithm \ref{alg:1} can be used to recover $c$ from $r$. 
\end{thm}
\begin{proof}
We prove this by induction on $m$. For $m=1$ and $m=2$, we have seen in Subsections \ref{ss:prs} and \ref{ss:prm2} that Algorithm \ref{alg:1} works for any degree $1\leq d\leq m(q-1)$. We assume that Algorithm \ref{alg:1} works for $m'=m-1$, and we will prove that this implies it also works for $m$. Let $c\in \PRM_d(m)$ and $f\in \fq[x_0,\dots,x_m]_d$ such that $\evp{m}(f)=c$. By Theorem \ref{T:recursive}, we have
$$
c=(u+v_{\xi,d},v),\; u\in \RM_{d-1}(m), v\in \PRM_d(m-1). 
$$
Now consider $e=(e_1,e_2)$, where $e_1\in \fq^{q^m}$, $e_2\in \fq^{p_{m-1}}$, with $\wt(e)=\wt(e_1)+\wt(e_2)< t_0$, and assume we receive the codeword $(r_1,r_2)=(u+v_{\xi,d}+e_1,v+e_2)$. Notice that, by Equation (\ref{eq:eta0}), we have 
$$
\wt(e_1)+\wt(e_2)< t_0= \frac{\etadm}{2} = \frac{\wt(\RM_d(m))}{2}+\frac{\sum_{i=1}^{m-\nu-1}\wt(\RM_d(m-i))+1}{2}.
$$
Therefore, either 
$$
\wt(e_1)< \frac{\wt(\RM_d(m))}{2}:=t_1,
$$
or 
$$
\wt(e_2)< \frac{\sum_{i=1}^{m-\nu-1}\wt(\RM_d(m-i))+1}{2}=\frac{\eta_d(m-1)}{2}:=t_2.
$$

{\bf First part:} the following argument corresponds to the first part of Algorithm \ref{alg:1}, which starts at Line \ref{alg:start:primeraparte} and finishes at Line \ref{alg:finish:primeraparte}. If $\wt(e_1)< t_1$, we can consider $(u+v_{\xi,d},f^0)=\DRM(m)(u+v_{\xi,d}+e_1)$ since $u+v_{\xi,d}\in \RM_d(m)$ (see Remark \ref{r:u+vRMd}) and $\wt(e_1)$ is lower than or equal to the error correction capability of $\RM_d(m)$. 

If $d\leq q-1$, i.e., $\nu=0$, we consider 
$$
f^0=f^0_{\leq d-1}+f^0_d,
$$
where $f^0_{\leq d-1}\in \fq[x_1,\dots,x_m]_{\leq d-1}$ and $f^0_d \in \fq[x_1,\dots,x_m]_d$. Then we have
$$
u=\eva{m}(f^0_{\leq d-1}), \; v=\evp{m-1}(f^0_d), \; c=\evp{m}(h^0_d(f^0)).
$$
Therefore, we have obtained $(c,f)=((u+v_{\xi,d},v),h^0_d(f^0))$.

If $d\geq q$, i.e., $\nu\geq 1$, as discussed in Subsection \ref{ss:prm2}, there might be several possible homogenizations of $f^0$. We consider the decomposition
$$
f^0=\fbado +\fgoodo=\fbado+(\fgoodo)_d +(\fgoodo)_{\leq d-1}.
$$
Let $\cgood:= \evp{m-1}((\fgoodo)_d)$. As a consequence of Lemma \ref{l:fdqminus1}, now we can consider $(\cbad,f'_{bad})=\DPRMq(m-1)(r_2-\cgood)$. By Corollary \ref{c:etaq-1} and the induction hypothesis, the error $e_2$ falls within the error correction capability of this decoder since $\wt(e_2)<t_0=\eta_d(m)/2=\eta_{d-(q-1)}(m-1)/2$. We recover $c=(c_1,\cgood+\cbad)=(u+v_{\xi,d},\cgood+\cbad)$. To recover $f$, first consider $(f'_{bad})_d\in \fq[x_1,\dots,x_m]_d$ the homogenization of $f'_{bad}$ up to degree $d$ as in Remark \ref{r:homogeneizar}. Let $g^0\in \fq[x_0,\dots,x_m]_{\leq d-1}$ be such that 
$$
g^0\equiv f^0-f'_{bad}-(\fgoodo)_d \equiv (\fgoodo)_{\leq d-1}+\fbad -f'_{bad} \bmod I(\A^m).
$$
Now we have
$$
f=h_d^0(g^0)+(f'_{bad})_d+(\fgoodo)_d.
$$
Since we do not know if $\wt(e_1)<t_1$ when decoding, in the first part of Algorithm \ref{alg:1} we follow the previous procedure to obtain a vector $c'$. Then we check if $\wt(r-c')<t$. If that is the case, we return $c'$ and the corresponding polynomial $f'$. If not, we proceed to the second part of Algorithm \ref{alg:1}.

{\bf Second part:} what follows corresponds to the second part of Algorithm \ref{alg:2}, which starts at Line \ref{alg:start:segundaparte}. Since the first part of Algorithm \ref{alg:1} has failed, it means that $\wt(e_1)\geq t_1$, and therefore $\wt(e_2)< t_2=\eta_d(m-1)/2$. By the induction hypothesis, we can consider $(v,g)=\DPRM(m-1)(v+e_2)$, and then perform the decoding (recall Remark \ref{r:dminRMYPRM})
$$
(u,f^0_{\leq d-1})=\DRMdd(m)(r_1-v_{\xi,d})=\DRMdd(m)(u+e_1).
$$
This provides $c=(u+v_{\xi,d},v)$. With respect to $f$, we just need to consider $f=h_d^0(f^0_{\leq d-1})+g$.

We have proved that the algorithm gives the correct output, and the fact that the algorithm finishes is clear by induction.
\end{proof}

Note that Algorithm \ref{alg:1} is recursive, since it is calling itself in Lines \ref{alg:linerecursive1} and \ref{alg:start:segundaparte}. A small example of how this procedure works has been given in Examples \ref{ex:differenthom2} and \ref{ex:differenthom3}. We will also show this in Example \ref{ex:patron}. Also note that, if at some point during Algorithm \ref{alg:1} we call $\DPRM(0)$, we understand $\wt(\PRM_d(0))=1$, and therefore $t=0$. Thus, $\DPRM(0)((r))$ returns $((r),rx_m^d)$ (see Subsection \ref{ss:prs}). 

\begin{algorithm}
\caption{}\label{alg:1}
\begin{algorithmic}[1] 
\Input Received word $r=c+e$, with $c\in \PRM_d(m)$ and $\wt(e)< t_0$. Decoders $\DRM(m-i)$ for $i=0,1,\dots,m-1$. 
\Output $(c,f)$, where $f\in \fq[x_0,\dots,x_m]_d$ and $\evp{m}(f)=c$. 
\State Let $t_0=\eta_d(m)/2$, $t=\wt(\PRM_d(m))/2$. 
\If{$t\leq 1$}
    \State Compute $f$ homogeneous of degree $d$ such that $\evp{m}(f)=r$.  \label{line:solution}
    \State \Return $(r,f)$
\EndIf
\State Define $r_1,r_2$ from $r=(r_1,r_2)$, where $r_1$ are the first $q^m$ coordinates and $r_2$ the last $q^{m-1}+\cdots+1$ coordinates. 

\If{$\DRM(m)(r_1)$ does not return an error}\label{alg:start:primeraparte} \Comment{{\bf First part}}
    \State $(c_1,f^0)\gets \DRM(m)(r_1)$
    \If {$d\leq q-1$} 
        \State $f=h^0_d(f^0)$
        \State $c=\evp{m}(f)$
        \If {$\wt(r-c)<t$}
            \State \Return $(c,f)$
        \EndIf
    \Else
        \State Compute the decomposition $f^0=(\fgoodo)_d+(\fgoodo)_{<d}+\fbado$.
        \State $\cgood=\evp{m-1}((\fgoodo)_d) $
        \If {$\DPRMq(m-1)(r_2-\cgood)$ does not return an error}
            \State $(\cbad,f_{bad}')\gets \DPRMq(m-1)(r_2-\cgood)$ \label{alg:linerecursive1}
            \State $g^0=f^0-f'_{bad}-(\fgoodo)_d \bmod I(\mathbb{A}^m)$.
            \State $f=h^0_d(g^0)+(f'_{bad})_d+(\fgoodo)_d$ 
            \State $c=(c_1,\cgood+\cbad)$
            \If {$\wt(r-c)<t$}
                \State \Return $(c,f)$
            \EndIf
        \EndIf
    \EndIf  
\EndIf \label{alg:finish:primeraparte} 
\State $(v,g)\gets \DPRM(m-1)(r_2)$ \Comment{{\bf Second part}} \label{alg:start:segundaparte} 
\State $v_{\xi,d}=(v,\xi^dv,\xi^{2d}v,\dots,\xi^{(q-2)d}v,0)$
\State $(u,f^0_{\leq d-1})\gets \DRMdd(m)(r_1-v_{\xi,d})$
\State $f=h_d^0(f^0_{\leq d-1})+g$ 
\State $c=(u+v_{\xi,d},v)$ 
\State \Return $(c,f)$
\end{algorithmic}
\end{algorithm}

\section{Analysis of the algorithm and improvements}\label{s:analysis}
Let $1\leq d \leq m(q-1)$ with $d-1=\nu (q-1)+\mu$, $0\leq \mu <q-1$. Algorithm \ref{alg:1} decodes any number of errors lower than $\eta_d(m)/2$. The previous algorithm known for PRM codes was given in \cite{decodingRMP}, and it was able to correct any number of errors lower than $\wt(\RM_d(m))/2$ (see \cite[Cor. 5.2]{decodingRMP}). Looking at Equation (\ref{eq:eta0}), it is clear that
$$
\eta_d(m)=\sum_{i=0}^{m-\nu-1} \wt(\RM_d(m-i))+1> \wt(\RM_d(m)),
$$
and therefore $\wt(\RM_d(m))/2< \eta_d(m)/2$, which means that Algorithm \ref{alg:1} can decode more errors in general. Moreover, since $\wt(\RM_d(m))<\wt(\RM_{d-1}(m))=\wt(\PRM_d(m))$, the algorithm from \cite{decodingRMP} decodes up to the error correction capability of $\PRM_d(m)$ only if $\wt(\RM_{d}(m))=\wt(\RM_{d-1}(m))-1$ (which only happens if and only if $d> (m-1)(q-1)$, i.e., $\nu= m-1$) and $\wt(\RM_{d-1}(m))=\wt(\PRM_d(m))$ is even. On the other hand, Algorithm \ref{alg:1} decodes up to the error correction capability of $\PRM_d(m)$ whenever $\mu=0$ or $\nu=m-1$ (see Lemma \ref{l:eta0}), independent of the parity of $\wt(\PRM_d(m))$, and in some additional cases depending on the parity of $\wt(\PRM_d(m))$. In Figures \ref{f:m2}, \ref{f:m2qlarge} and \ref{f:mlarge} we show the value of $T_0/T$, where
$$
T=\left\lfloor \frac{\wt(\PRM_d(m))-1}{2}\right\rfloor, \; T_0=\left\lfloor \frac{\eta_d(m)-1}{2}\right\rfloor,
$$
i.e., $T_0/T$ is the quotient of the number of errors that Algorithm \ref{alg:1} can correct with respect to the error correction capability of $\PRM_d(m)$. For example, for $q=3$ and $m=2$, Algorithm \ref{alg:1} can correct up to the error correction capability of $\PRM_d(2)$ for any degree for which $\PRM_d(2)$ is non trivial. As noted before, for $\mu=0$ or $\nu=m-1$, we always have $T_0/T=1$, i.e., we correct up to the error correction capability of $\PRM_d(m)$. Also note that, for $q$ large, the error correction capability of the decoding algorithm is very close to that of $\PRM_d(m)$ for most degrees (see Figure \ref{f:m2qlarge}). 

For some degrees, PRM codes can be seen as Grassmann codes, for which some decoding algorithms have been studied \cite{beelen_majority_decoding_grassmann,kroll_PD_sets_binary_RM_schubert}. However, Algorithm \ref{alg:1} can correct more errors in general than these decoders for the corresponding PRM codes. For example, in \cite[Rem. 4.7]{beelen_majority_decoding_grassmann}, the authors mention that their algorithm decodes up to half the minimum distance for $q=2$, and roughly up to $\wt(\PRM_d(1))/4$ for large $q$. Algorithm \ref{alg:1} decodes always up to the error correction capability of $\PRM_d(1)$ for any $q$.

\begin{figure}
\caption{$T_0/T$ as a function of $1\leq d \leq m(q-1)$, for $m=2$ and $q$ small.} \label{f:m2}
    \begin{center}
        \resizebox{0.45\linewidth}{!}{\input{GF3m2.pgf}}
        \resizebox{0.45\linewidth}{!}{\input{GF4m2.pgf}}
        \resizebox{0.45\linewidth}{!}{\input{GF5m2.pgf}}
        \resizebox{0.45\linewidth}{!}{\input{GF7m2.pgf}}
    \end{center}
\end{figure}

\begin{figure}
\caption{$T_0/T$ as a function of $1\leq d \leq m(q-1)$, for $m=2$ and $q$ large.} \label{f:m2qlarge}
    \begin{center}
        \resizebox{0.45\linewidth}{!}{\input{GF16m2.pgf}}
        \resizebox{0.45\linewidth}{!}{\input{GF64m2.pgf}}
        \resizebox{0.45\linewidth}{!}{\input{GF256m2.pgf}}
        \resizebox{0.45\linewidth}{!}{\input{GF1024m2.pgf}}
    \end{center}
\end{figure}

\begin{figure}
\caption{$T_0/T$ as a function of $1\leq d \leq m(q-1)$, for $m=5$.}  \label{f:mlarge}
    \begin{center}
        \resizebox{0.45\linewidth}{!}{\input{GF3m5.pgf}}
        \resizebox{0.45\linewidth}{!}{\input{GF4m5.pgf}}
        \resizebox{0.45\linewidth}{!}{\input{GF5m5.pgf}}
        \resizebox{0.45\linewidth}{!}{\input{GF7m5.pgf}}
    \end{center}
\end{figure}

If we have an error $e$ with $T_0< \wt(e) \leq T$, Algorithm \ref{alg:1} may still be able to decode, but it requires some additional exception handling when programming (see Algorithm \ref{alg:2}). For example, in Algorithm \ref{alg:1}, in Line \ref{line:solution}, when we are in the context of Theorem \ref{T:alg}, if we get into this part of the algorithm, there should always be a solution for that system, but when the number of errors is higher, there might not be a solution.  

\begin{algorithm}
\caption{}\label{alg:2}
\begin{algorithmic}[1] 
\Input Received word $r=c+e$, with $c\in \PRM_d(m)$ and $\wt(e)< t_0$. Decoders $\DRM(m-i)$ for $i=0,1,\dots,m-1$. 
\Output $(c,f)$, where $f\in \fq[x_0,\dots,x_m]_d$ and $\evp{m}(f)=c$. 
\State Let $t_0=\eta_d(m)/2$, $t=\wt(\PRM_d(m))/2$. 
\If{$t\leq 1$}
    \If {there is $f$ such that $\evp{m}(f)=r$} 
        \State \Return $(r,f)$
    \Else 
        \State \Return Error
    \EndIf
\EndIf
\State Define $r_1,r_2$ from $r=(r_1,r_2)$, where $r_1$ are the first $q^m$ coordinates and $r_2$ the last $q^{m-1}+\cdots+1$ coordinates. 

\If{$\DRM(m)(r_1)$ does not return an error}\label{alg:start:primeraparte2} \Comment{{\bf First part}}
    \State $(c_1,f^0)\gets \DRM(m)(r_1)$
    \If {$d\leq q-1$} 
        \State $f=h^0_d(f^0)$
        \State $c=\evp{m}(f)$
        \If {$\wt(r-c)<t$}
            \State \Return $(c,f)$
        \EndIf
    \Else
        \State Compute the decomposition $f^0=(\fgoodo)_d+(\fgoodo)_{<d}+\fbado$.
        \State $\cgood=\evp{m-1}((\fgoodo)_d) $
        \If {$\DPRMq(m-1)(r_2-\cgood)$ does not return an error}
            \State $(\cbad,f_{bad}')\gets \DPRMq(m-1)(r_2-\cgood)$
            \State $g^0=f^0-f'_{bad}-(\fgoodo)_d \bmod I(\mathbb{A}^m)$.
            \State $f=h^0_d(g^0)+(f'_{bad})_d+(\fgoodo)_d$ 
            \State $c=(c_1,\cgood+\cbad)$
            \If {$\wt(r-c)<t$}
                \State \Return $(c,f)$
            \EndIf
        \EndIf
    \EndIf  
\EndIf \label{alg:finish:primeraparte2} 
\If{$\DPRM(m-1)(r_2)$ does not return an error} \label{alg:start:segundaparte2} \Comment{{\bf Second part}}  
    \State $(v,g)\gets \DPRM(m-1)(r_2)$ 
    \State $v_{\xi,d}=(v,\xi^dv,\xi^{2d}v,\dots,\xi^{(q-2)d}v,0)$
    \State $(u,f^0_{\leq d-1})\gets \DRMdd(m)(r_1-v_{\xi,d})$
    \State $f=h_d^0(f^0_{\leq d-1})+g$ 
    \State $c=(u+v_{\xi,d},v)$ 
    \State \Return $(c,f)$
\Else
    \State \Return Error
\EndIf
\end{algorithmic}
\end{algorithm}

An example of how Algorithm \ref{alg:2} can decode some errors $e$ with $T_0< \wt(e) \leq T$ is the following. If all the errors are concentrated in the first $q^m$ coordinates, but there are no errors in the rest, then the first part of Algorithm \ref{alg:2} will fail to return $(c,f)$, but the second one will correctly compute the last $p_m-q^m$ coordinates of the vector sent, and then we can use $\DRMdd(m)$ to recover the first $q^m$ coordinates, since this decoder can correct any number of errors lower than or equal to $T$. A generalization of this idea is given in the next result.  
\begin{prop}\label{p:errorpattern}
Let $d-1=\nu (q-1)+\mu$, $d\leq m(q-1)$, $0\leq \mu  <q-1$, and let $e$ be an error vector with $\wt(e)<t$. Assume there is some $i\leq m$ such that
\begin{enumerate}
    \item \label{cond:1}$\wt((e_Q)_{Q\in \{0\}^{m-j}\times P^j})<\frac{\wt(\PRM_d(j))}{2}$, for all $i< j \leq m$,
    \item \label{cond:2}$\wt((e_Q)_{Q\in \{0\}^{m-i}\times P^i})<\frac{\eta_d(i)}{2}$.
\end{enumerate}
Then Algorithm \ref{alg:2} can successfully correct $e$.
\end{prop}
\begin{proof} 
We argue by induction on $m$. For $m=1$, we can have $i=0$, in which case we are stating that the algorithm from Subsection \ref{ss:prs} works; and we can have $i=1$, which is again the case of Subsection \ref{ss:prs} since $\eta_d(1)=\wt(\PRS_d)$. We assume the result is true for $m'=m-1$, and we prove it now for $m$. First, we enter the first part of Algorithm \ref{alg:2}. If the algorithm finishes in this part, then the decoding is correct since we always check that the difference between the received codeword and the decoded codeword is lower than $t$. If we do not finish in the first part and we move to the second part, the first thing we do is to consider $\DPRM(m-1)(r_2)$. By the induction hypothesis (note that $e_2=(e_Q)_{Q\in \{0 \}\times P^{m-1}}$ also satisfies the corresponding conditions), we obtain $(v,g)=\DPRM(m-1)(r_2)$. Now we consider $\DRMdd(m)(r_1-v_{\xi,d})$, which will also successfully give $(u,f_{\leq d-1})$ (taking into account that $\wt(e)<t$ and Remark \ref{r:dminRMYPRM}). Thus, we obtain $(c,f)$ following Algorithm \ref{alg:2}. 
\end{proof}

\begin{ex}\label{ex:patron}
We continue with the setting from Example \ref{ex:P2}. For $d=3$ and $m=2$, we have that $\PRM_3(2)$ is a $[21,10,8]$ code. This code can correct $T=3$ errors. One can check that $\RM_3(2)$ and $\RS_3$ have parameters $[16,10,4]$ and $[4,4,1]$, respectively. Thus, $\eta_3(2)=6$ and Algorithm \ref{alg:1} can correct $T_0=2$ errors. This is a particularly unfavorable case for Algorithm \ref{alg:1}, since it is the only degree for which $T_0/T<1$ for $q=4$ and $m=2$ (see Figure \ref{f:m2}). However, using Algorithm \ref{alg:2}, we can still correct $3$ errors in some cases. Let
$$
\begin{aligned}
&c=\evp{2}(x_0^3+x_1^3+x_2^3)=(1, 1, 1, 0, 0, 1, 1, 1, 0, 0, 1, 1, 1, 0, 0, 1, 0, 0, 0, 1, 1), \\
&e=(a, a + 1, 0, 0, 0, 1, 0, 0, 0, 0, 0, 0, 0, 0, 0, 0, 0, 0, 0, 0, 0),\\
&r=c+e=(a + 1, a, 1, 0, 0, 0, 1, 1, 0, 0, 1, 1, 1, 0, 0, 1, 0, 0, 0, 1, 1).
\end{aligned}
$$
Note that
$$
\wt((e_Q)_{Q\in \{0\}\times P^1})=0<t_0=3=\eta_3(2)/2,\;  \wt((e_Q)_{Q\in P^2})=3=T.
$$
Thus, by Proposition \ref{p:errorpattern}, we can apply Algorithm \ref{alg:2} to recover $c$ from $r$. Let $c=(c_1,c_2)$ and $r=(r_1,r_2)$ as in the proof of Theorem \ref{T:alg}. In particular, 
$$
r_1=(a + 1, a, 1, 0, 0, 0, 1, 1, 0, 0, 1, 1, 1, 0, 0, 1), \; r_2=(0, 0, 0, 1, 1).
$$
In the first part of Algorithm \ref{alg:2}, we compute $D^{\mathbb{A}}_3(2)(r_1)$, which returns an error. This is due to the fact that
$$
\wt((e_Q)_{Q\in \{1\}^\times \fq^2})=3>2= \wt(\RM_3(2))/2.
$$
Now we go to the second part of Algorithm \ref{alg:2}, and we consider $(v,g)=D^{\mathbb{P}}_3(1)(r_2)$. To do this, we perform Algorithm \ref{alg:2} with $m'=1$ (or the algorithm explained in Subsection \ref{ss:prs}). Since $\wt(\PRS_3)/2=1$, $D^{\mathbb{P}}_3(1)$ directly returns the received vector, and we have that 
$$
D^{\mathbb{P}}_3(1)(r_2)=(r_2,g)=((0, 0, 0, 1, 1),x_1^3+x_2^3),
$$
where $x_1^3+x_2^3$ can be obtained by solving a linear system of equations. Note that $\xi^3=a^3=1$, and we obtain
$$
v_{\xi,3}=(r_2)_{\xi,3}=(r_2,r_2,r_2,0)=(0, 0, 0, 1, 1, 0, 0, 0, 1, 1, 0, 0, 0, 1, 1, 0).
$$
Now we denote $(u,f^0_{\leq 2})= D^{\A}_2(2)(r_1-v_{\xi,d})$. Since $\RM_2(2)$ has parameters $[16,6,8]$, it can correct $\wt((e_Q)_{Q\in \{1\}^\times \fq^2})=3$ errors, which implies that this last decoding does not fail and correctly returns
$$
u=(1, 1, 1, 1, 1, 1, 1, 1, 1, 1, 1, 1, 1, 1, 1, 1), \; f^0_{\leq 2}=1.
$$
Then we have
$$
c=(u+v_{\xi,d},v)=(1, 1, 1, 0, 0, 1, 1, 1, 0, 0, 1, 1, 1, 0, 0, 1, 0, 0, 0, 1, 1),
$$
and 
$$
f=h_3^0(f^0_{\leq 2})+g=x_0^3+x_1^3+x_2^3.
$$
\end{ex}

Additionally, some decoders for RM codes are also known to decode more errors than the error correction capability, as long as these errors are in general position (see  \cite{brasCorrectionCapabilityBMS}). By using such decoders with Algorithm \ref{alg:2}, again we may also be able to correct more errors than $T_0$. 

The subfield subcodes of PRM codes (that is, the intersection of a PRM code with $\F_{q'}^n$, where $\F_{q'}\subset \fq$) have been studied in \cite{sanjoseSSCPRM,sanjoseRecursivePRM}. They have been shown to have good parameters, and they can be used for practical applications \cite{sanjoseSSCPRS,sanjoseHullvariationPRM,sanjoseHullsPRM}. Since these codes are subcodes of PRM codes, they can also be decoded with Algorithms \ref{alg:1} and \ref{alg:2}. The usual problem with this approach is that the subfield subcode is defined over a smaller field $\F_{q'}$, and the decoding might require operations over the bigger field $\fq$. However, in \cite[Cor. 3]{sanjoseRecursivePRM}, it is shown that, if $d$ is a multiple of $(q-1)/(q'-1)$, then the recursive construction from Theorem \ref{T:recursive} also works for the subfield subcodes. Thus, Algorithm \ref{alg:1} can decode the subfield subcodes for these degrees using only operations over $\F_{q'}$, if we consider appropriate decoders for the subfield subcodes of affine RM codes. 

\section{Complexity}\label{s:complexity}
Since the decoder $\DPRM(m)$ we have presented uses a decoder $\DRM(m)$ for affine RM codes, the complexity will depend on that of $\DRM(m)$. Let $\theta^{\A}(m)$ be the order of complexity of $\DRM(m)$. For example, the majority voting decoding algorithm \cite{duursmaMajority, fengraoMajority} has complexity $O((n^\A)^3)$, i.e., $\theta^{\A}(m)=(n^\A)^3$ in this case, where $n^\A=q^m$. It is usual to express the order of complexity in terms of the length of the code. Note that
\begin{equation*}\label{eq:fraccionlength}
\frac{p_m}{n^\A}=\frac{q^{m+1}-1}{q^m(q-1)}=\frac{q^m+q^{m-1}+\dots +1}{q^m}=\frac{1-(1/q)^{m+1}}{1-1/q}=\frac{q-(1/q)^m}{q-1} \xrightarrow[m \to \infty]{} \frac{q}{q-1}.
\end{equation*}
This means that the length of $\PRM_d(m)$ is about $q/(q-1)$ times higher than that of $\RM_d(m)$. This does not change the order of complexity, and therefore we can use both $p_m$ or $n^\A$ to express the order of complexity of $\DPRM(m)$. 

\begin{prop}
The worst case complexity of Algorithms \ref{alg:1} and \ref{alg:2} is $O(\comp)$.
\end{prop}
\begin{proof}
We can assume that $\comp >> \theta^\A(m-i)$, for any $i=1,\dots,m-1$, and we also assume that the cost of the rest of the operations (e.g., homogenizing a polynomial or adding polynomials) is negligible. Thus, the order of complexity is given by $\comp$, since in the worst case we use both $\DRM(m)$ and $\DRMdd(m)$. 
\end{proof}

As we explained in Section \ref{s:analysis}, the decoding algorithm from \cite{decodingRMP} decodes any number of errors lower than $\wt(\RM_d(m))/2$. In that setting, Algorithm \ref{alg:1} always finishes in the first part and only requires using $\DRM(m)$. In this sense, Algorithm \ref{alg:1} generalizes the algorithm from \cite{decodingRMP} not only because it can correct more errors, but also because their complexity is the same when both algorithms can be applied. 

\section{Conclusion}\label{s:conclusion}
We have shown how to decode PRM codes if the number of errors is lower than $\eta_d(m)/2$, or if the error satisfies certain conditions. We have also studied the computational complexity of the corresponding algorithms. Future work in this direction may include improvements in the number of errors that can be corrected and possible generalizations to list decoding.

\bibliographystyle{abbrv}

\end{document}